\documentclass[submission,copyright,creativecommons]{eptcs}
\providecommand{\event}{TARK 2025} 
\usepackage{thmtools,thm-restate}
\usepackage{amsmath,amsthm} 
\usepackage{amsthm}
\usepackage{thmtools,thm-restate}
\usepackage{xcolor}
\usepackage{iftex}
\usepackage{amsmath}
\usepackage{amssymb}

\ifpdf
  \usepackage{underscore}         
  \usepackage[T1]{fontenc}        
\else
  \usepackage{breakurl}           
\fi

\theoremstyle{definition}
\newtheorem{theorem}{Theorem}[section]

\newtheorem{proposition}[theorem]{Proposition}

\newtheorem{definition}[theorem]{Definition}

\newtheorem{lemma}[theorem]{Lemma}

\newtheorem{example}[theorem]{Example}

\newtheorem{remark}[theorem]{Remark}

\title{Common $p$-Belief with Plausibility Measures\\ Extended Abstract}
\author{
Eric Pacuit \qquad\qquad Leo Yang
\institute{Department of Philosophy \\ University of Maryland}
\email{\quad epacuit@umd.edu \quad\qquad leoyang@umd.edu}
}
\def\titlerunning{Common $p$-Belief with Plausibility Measures: Extended Abstract}
\def\authorrunning{E. Pacuit and L. Yang}

\begin{document}

\maketitle

\section{Introduction}\label{introduction}

In his seminal 1976 paper, Robert Aumann proved a fascinating result, later called the ``Agreeing to Disagree'' theorem \cite{Aumann1976}.   Suppose that two agents share the same prior probability distribution and update this probability by conditioning on different  private information. 
Aumann showed that if their posterior probabilities of an event $E$ become common knowledge---meaning each agent knows the other's posterior, knows that the other knows it, and so on---then they must assign identical posterior probabilities to $E$. In other words, two Bayesian agents who share a common prior cannot agree (in the sense of having common knowledge) to disagree (in the sense of holding different posterior probabilities).  This result has important implications across multiple disciplines: it underlies ``no-trade" theorems in economics \cite{MilgromStokey1982,SebeniusGeanakoplos1983}, helps to provide an epistemic characterization of Nash equilibria in games with more than two players \cite{AumannBrandenburger19995}, and reveals  subtle issues that arise when rational agents exchange information  \cite{GeanakoplosPolemarchakis1982,ParikhKrasucki1990,Menager2006}.

The literature on Aumann's theorem has focused on identifying which assumptions are necessary to prove the theorem \cite{RubinsteinWolinsky1990,Geanakoplos1994,Lederman2015}. Various generalizations have emerged, exploring scenarios where, for instance, agents have ambiguous beliefs \cite{HalpernKets2014,HalpernKets2015,VergopoulosBillot2024} or use  learning rules other than classical Bayesian conditioning \cite{GilboaSamuelsonSchmeidler2022,HildJeffreyRisse1997}.   This paper is inspired by three important generalizations of Aumann's theorem:

\begin{enumerate}

\item There are several  {\em qualitative} versions of Aumann's theorem. 
One approach focuses on reformulating Aumann's insights in multi-agent logics of belief \cite{DegremontRoy2012,Demey2014}. A second approach, developed by Cave \cite{Cave1983}, Bacharach \cite{Bacharach1985}, and Samet \cite{Samet2010}, reconceptualizes the core problem by shifting focus from probabilities to {\em abstract decisions}. In Aumann's original theorem, the posterior probability of an agent for a fixed event can be viewed as a random variable mapping states to real numbers in the  $[0, 1]$ interval. The core idea is to replace these random variables with decision functions mapping worlds to elements of a set $\mathcal{D}$ of ``decisions''. \footnote{The resulting qualitative agreement theorem establishes that ``like-minded'' agents cannot have common knowledge of different decisions, with the concept of being ``like-minded'' formalized using a variant of Savage's sure-thing principle: For any pair of agents $i$ and $j$,  if $i$ knows that $j$ is at least as knowledgeable as she is, and also knows that $j$’s decision is $d$, then her decision is also $d$ \cite[p. 171]{Samet2022}.}

\item A second direction of research extends Aumann's theorem to generalizations of  classical Bayesian models, including conditional probability structures \cite{Tsakas2018}, lexicographic probability systems \cite{BachPerea2013,BachCabessa2023}, and  imprecise probabilities \cite{ZhangLiuSeidenfeld2018}. By reformulating the agreement theorem within these more general models of beliefs, researchers have shown that Aumann's core insight---that rational agents cannot agree to disagree---remains robust even when relaxing classical probability assumptions.   

\item The third line of research generalizes Aumann's theorem by replacing knowledge with the weaker notion of {\em $p$-belief}.  An agent $p$-believes an event $E$ when assigning a probability of at least $p$ to $E$. In this setting, we can ``approximate" common knowledge with {\em common $p$-belief}.   Monderer and Samet \cite{MondererSamet1989}, later refined by Neeman \cite{Neeman1996}, proved that if agents sharing a common prior have common $p$-belief of their posteriors of some event $E$, then these posteriors must be ``close'', in the sense that they cannot differ by more than $1-p$.   This elegant result shows that if the agents are close to having common knowledge of their posteriors of an event $E$, then these posteriors must also be close.  Then, Aumann's  agreement theorem can be viewed as the special case when $p=1$.

\end{enumerate}

This paper lies in the intersection of these three approaches, aiming to identify the minimal assumptions that any belief model must satisfy to prove a version of the Monderer-Samet-Neeman  agreement theorem.  To accomplish this, we use Halpern's \cite{FriedmanHalpern1995,Halpern2001} (conditional) {\em plausibility measures} to represent the agents' beliefs. These measures provide a very general framework that encompasses many existing formal models of beliefs, including probability measures, possibility measures, ranking functions, and Dempster-Shafer belief functions (see \cite{Halpern2017} for a discussion).  

Our generalized Monderer-Samet-Neeman theorem using conditional plausibility measures provides a unifying perspective that clarifies which properties of belief  are essential for demonstrating that common belief of the posteriors of some event implies these posteriors must be close. This deepens our theoretical understanding of agreement theorems while extending the applicability of the Monderer-Samet-Neeman Theorem to settings where classical probabilistic assumptions may not hold.

This extended abstract is organized as follows. Section \ref{background-Aumann} provides an overview of Aumann's classic agreeing to disagree theorem. We introduce (common) $p$-belief and discuss the  Monderer-Samet-Neeman generalization of Aumann's theorem in Section \ref{sec-common-p-belief}. Section \ref{epistemic-plausibility-models} introduces conditional plausibility measures   \cite{FriedmanHalpern1995,Halpern2001} and states our generalizations of the Monderer-Samet-Neeman agreement theorem.  The  proofs of the results are available in the full version of the paper.  We conclude in Section \ref{conclusion} with a brief discussion of some applications of our results.

\section{Background: Aumann's Agreement Theorem}\label{background-Aumann}
In this section, we briefly introduce the formal definitions needed to state and prove Aumann's agreeing to disagree theorem.  We begin with models representing the knowledge and beliefs of a group of agents.

\begin{definition}[Epistemic Probability Model]\label{epistemic-probability-model} Suppose that $\mathcal{A}$ is a set of agents.  An {\bf epistemic probability model} for $\mathcal{A}$ is a  tuple
$$\langle W,(\Pi_{i})_{i\in\mathcal{A}}, \mathcal{F},(P_i)_{i\in\mathcal{A}} \rangle,$$
where 
    \begin{enumerate}
        \item $W$ is a non-empty set of states; 
        \item For each $i\in\mathcal{A}$,  $\Pi_i$ is a partition of $W$, where for each $w\in W$ and $i\in\mathcal{A}$, we write $\Pi_i(w)$ for the unique element of $\Pi_i$ containing $w$, called agent $i$'s \textbf{information set at $w$};\footnote{A partition of $W$ is a collection of pairwise disjoint sets whose union is $W$.}
        \item $\mathcal{F}\subseteq \wp(W)$ is a \textbf{$\sigma$-algebra} on $W$, meaning that $\mathcal{F}$ is closed under complement and countable union, such that for all $w\in W$ and $i\in\mathcal{A}$, $\Pi_i(w)\in \mathcal{F}$; 
        \item For each $i\in\mathcal{A}$, $P_i$ is a probability measure on $\langle W,\mathcal{F} \rangle$. 
    \end{enumerate}
An epistemic probability model has a \textbf{common prior} if $P_i=P_j$ for all $i,j\in\mathcal{A}$. To simplify notation, we write $\langle W,(\Pi_{i})_{i\in\mathcal{A}}, \mathcal{F},P \rangle$ when $P$ is the common prior.
\end{definition} 
\begin{remark}
    We could also allow a different $\sigma$-algebra $\mathcal{F}_i$ for each $P_i$, but this complexity is unnecessary for the purpose of this paper. 
\end{remark}

Given a state $w$ and an agent $i\in\mathcal{A}$, the set $\Pi_i(w)$ represents $i$'s information at state $w$.  With this interpretation in mind, we can now define the agents' knowledge and posterior beliefs.

\begin{definition}[Knowledge]\label{knowledge}
Suppose that   $\langle W,(\Pi_{i})_{i\in\mathcal{A}},\mathcal{F},(P_i)_{i\in\mathcal{A}} \rangle$  is an epistemic probability model. The {\bf knowledge operator} for agent $i\in\mathcal{A}$ is the function $K_i:\wp(W)\rightarrow\wp(W)$ where: for all $E\subseteq W$,
$$K_i(E)=\{w\ |\ \Pi_i(w)\subseteq E\}.$$
If $w\in K_i(E)$ we say that agent $i$ knows that $E$ in $w$.\footnote{It  is well known that each $K_i$ satisfies the $\mathbf{S5}$ axioms: For all events $E, F\subseteq W$, $K_i(E\cap F)=K_i(E)\cap K_i(F)$, $K_i(E)\subseteq E$, $K_i(E)\subseteq K_i(K_i(E))$, and $\overline{K_i(E)}\subseteq K_i(\overline{K_i(E)})$ (where $\overline{X}$ denotes the complement of $X$ in $W$).}
\end{definition}

\begin{definition}[Posterior Belief]\label{posterior-belief} Given an epistemic probability model $\langle W, (\Pi_i)_{i\in\mathcal{A}}, \mathcal{F},(P_i)_{i\in\mathcal{A}}\rangle$. 
The {\bf posterior belief for $i$ at state $w$} is the (possibly partial) function $P_{i,w}:\wp(W)\rightarrow [0,1]$ where: for all  $E\subseteq W$,
$$P_{i,w}(E)=P_i(E\ |\ \Pi_i(w)) = \frac{P_i(E\cap \Pi_i(w))}{P_i(\Pi_i(w))}.$$ 
\noindent Note that $P_{i,w}(E)$ is undefined if $E\cap F\notin\mathcal{F}$, $\Pi_i(w)\notin\mathcal{F}$, or $P_i(\Pi_i(w))=0$. 
\end{definition}

Suppose that $\langle W, (\Pi_i)_{i\in\mathcal{A}},\mathcal{F},(P_i)_{i\in\mathcal{A}}\rangle$ is an epistemic-probability model. We write $K(E)$ for the event that  everyone knows $E$, and $C(E)$ for the event that $E$ is common knowledge.\footnote{These definitions can be relativized to any subset of agents. For $G\subseteq \mathcal{A}$, let $K_G(E)$ be the event that everyone in $G$ knows $E$, and $C_G(E)$ be the event that $E$ is common knowledge among everyone in $G$. To keep notation minimal, we do not pursue this more general approach in this paper.} Formally, {\em everyone knows} is a function $K:\wp(W)\rightarrow \wp(W)$ where: for all $E\subseteq W$,
$$K(E)=\bigcap_{i\in \mathcal{A}}K_i(E).$$    
Common knowledge of $E$ is the event where everyone knows $E$ and this fact is completely transparent to all agents. This can be made  precise using the concept of a self-evident event:

\begin{definition}[Self-Evident Event]\label{self-evident}
Suppose that $\langle W, (\Pi_i)_{i\in\mathcal{A}},\mathcal{F},(P_i)_{i\in\mathcal{A}}\rangle$ is an epistemic probability model.   An event $E\subseteq W$ is said to be {\bf self-evident for $i$} provided that $E\subseteq K_i(E)$. 
An event is {\bf self-evident} if it is self-evident for all $i\in \mathcal{A}$. 
\end{definition}
Thus, the event $E$ is self-evident when $E$ is closed with respect to the agents' information partitions: for all $i\in\mathcal{A}$ and all $w\in W$, if $w\in E$, then $\Pi_i(w)\subseteq E$.

\begin{definition}[Common Knowledge]\label{common-knowledge} Suppose that $\langle W, (\Pi_i)_{i\in\mathcal{A}},\mathcal{F},(P_i)_{i\in\mathcal{A}}\rangle$ is an epistemic probability model. 
The {\bf common knowledge operator} is the function $C:\wp(W)\rightarrow\wp(W)$ such that, for each $E\subseteq W$,
$$C(E)=\{w\ |\ \mbox{there is a self-evident set $Z$ such that $w\in Z\subseteq E$}\}.$$
 
\end{definition}  
An event $E$ is common knowledge provided that there is a true self-evident event that implies $E$. 
We can now state Aumann's agreeing to disagree theorem. 

\begin{restatable}[Agreeing to Disagree Theorem \cite{Aumann1976}]{theorem}{aumannagreement}\label{agreement-thm}
    Suppose that $\langle W,(\Pi_i)_{i\in\mathcal{A}},\mathcal{F},P \rangle$ is an epistemic probability model with a common prior, and for each $i\in\mathcal{A}$,  $r_i\in [0,1]$.  For any event $E\subseteq W$, if 
    $$C(\{w\mid P_{i,w}(E)=r_i\mbox{ for all }i\in\mathcal{A}\})\neq\emptyset$$
    then $r_i=r_j$ for all $i, j\in\mathcal{A}$. 
\end{restatable}

\section{Common $p$-Belief}\label{sec-common-p-belief}

In this section, we present the generalization of Aumann's agreement theorem (Theorem \ref{agreement-thm}) first proven by Monderer and Samet \cite{MondererSamet1989} and later refined by Neeman \cite{Neeman1996}.    Before stating the theorem, we introduce several key definitions.

\begin{definition}[$p$-Belief]\label{p-belief}
Suppose that $\langle W, (\Pi_i)_{i\in\mathcal{A}},\mathcal{F},(P_i)_{i\in\mathcal{A}}\rangle$ is an epistemic probability model and  $p\in (0,1]$. 
Agent $i$'s {\bf $p$-belief} is the function $B^p_i : \mathcal{F}\to \wp(W)$ 
where: for all $E\subseteq W$,
$$B^p_i(E)=\{w\mid P_{i,w}(E)\geq p\}.$$
If $w\in B_i^p(E)$ we say that $i$ believes $E$ to degree at least $p$ in $w$. 
\end{definition}
\begin{remark}
    Note that $w\notin B^p_i(E)$ if $P_{i,w}(E)$ is undefined.   Moreover, we assume that for all $E\in\mathcal{F}$ and $i\in\mathcal{A}$ and $p\in (0, 1]$, $B_i^p(E)\in\mathcal{F}$.    Note that we can prove that $\mathcal{F}$ is closed under $B_i^p$  if we assume that $W$ is finite or countable.
\end{remark}

The definition of common $p$-belief closely follows the definition of common knowledge (Definitions \ref{self-evident}  and \ref{common-knowledge}).  We first define the \textbf{mutual $p$-belief operator} $B^p: \wp(W) \to \wp(W)$ as follows: for all $E\subseteq W$,
$$B^p(E)=\bigcap_{i\in\mathcal{A}} B^p_i(E).$$ 
Mutual $p$-belief of $E$ is the event where all agents in $\mathcal{A}$ believe  $E$ with probability at least $p$.  Common $p$-belief of $E$ is the event where everyone $p$-believes $E$ and this fact is ``partially transparent" to all agents.

\begin{definition}[$p$-Self-Evident Event]\label{p-self-evident}
Suppose that $\langle W, (\Pi_i)_{i\in\mathcal{A}},\mathcal{F},(P_i)_{i\in\mathcal{A}}\rangle$ is an epistemic probability model. An event $E\subseteq W$ is {\bf $p$-self-evident for agent $i$} if $E\subseteq B^p_i(E)$.  An event is {\bf $p$-self-evident} if it is $p$-self-evident for all  $i\in\mathcal{A}$. 
\end{definition}

\begin{definition}[Common $p$-Belief]\label{common-p-belief}
    Suppose that $\langle W, (\Pi_i)_{i\in\mathcal{A}},\mathcal{F},(P_i)_{i\in\mathcal{A}}\rangle$ is an epistemic probability model.   For $p\in (0, 1]$, the {\bf common $p$-belief operator} is the function  $CB^p:\wp(W)\rightarrow \wp(W)$ where: for all $E\subseteq W$, 
    $$CB^p(E)=\{w\ |\ \mbox{there is a $p$-self-evident event $Z$ such that } w\in Z\subseteq B^p(E)\}.$$
\end{definition}
So, an event $E$ is common $p$-believed when there is a true $p$-self-evident event that implies everyone $p$-believes $E$.  Next, we state a useful lemma that is needed to prove the Moderer–Samet–Neeman generalization of the agreement theorem.

\begin{restatable}{lemma}{usefullemmaone}
\label{UsefulLemma1}
    Suppose that  $\langle W,(\Pi_i)_{i\in\mathcal{A}}, \mathcal{F},(P_i)_{i\in\mathcal{A}}\rangle$ is an epistemic probability model.  Let $E\subseteq W$ be an event, $p\in (0, 1]$, and $r_i\in [0,1]$ for each $i\in \mathcal{A}$. Then, for each $i\in \mathcal{A}$: 
   \begin{enumerate}
   \item $\Pi_i$ partitions $B_i^p(E)$, i.e., $B_i^p(E)=\bigcup \{X\mid X\in \Pi_i\mbox{ and } X\cap B_i^p(E)\neq\emptyset\}$.
    \item $CB^p(E)\subseteq B^p_i CB^p(E)$.
    \item If $B^p_i(E)\neq\emptyset$, then $P_i(E\mid B^p_i(E))\geq p$. 
    \item Let $X=\{w \mid P_{i,w}(E)=r_i\text{ for all }i\in\mathcal{A}\}$. If $ CB^p(X)\neq\emptyset$, then $P_i(E\mid B^p_i CB^p(X))=r_i$.
\end{enumerate}
\end{restatable}

\begin{proof}[Proof (Sketch)]
    \begin{enumerate}
    \item For any $E\subseteq W$ and $w, w'\in W$, if $w'\in \Pi_i(w)$, then $P_i(E\mid \Pi_i(w))=P_i(E\mid \Pi_i(w'))$.  Thus, if $X\in\Pi_i$ and $X\cap B_i^p(E)\neq \emptyset$, then $X\subseteq B_i^p(E)$.    This immediately implies that $B_i^p(E)=\bigcup\{X\mid X\in\Pi_i\mbox{ and } X\cap B_i^p(E)\neq\emptyset\}$.
    \item Suppose that $w\in CB^p(E)$.  Then there is a $p$-self-evident event $Z$   such that $w\in Z\subseteq B^p(E)$.   Since $Z$ is $p$-self-evident, we have $Z\subseteq B_i^p(Z)$ and so, $P_{i}(Z\mid \Pi_i(w))\geq p$.   Clearly, $Z\subseteq CB^p(E)$, and so, $P_{i}(CB^p(E)\mid \Pi_i(w))\geq P_{i}(Z\mid \Pi_i(w))\geq p$.  Thus,  $w\in B^p_iCB^p(E)$.
    \item By part 1., $B_i^p(E)=\bigcup\{X\mid X\in\Pi_i\mbox{ and } X\cap B_i^p(E)\neq\emptyset\}$.   First of all, note that since $P_i(\Pi_i(w))>0$ for all $w\in W$, we have that if $B_i^p(E)\neq \emptyset$, then $P_i(B_i^p(E)) > 0$.   Let $\mathcal{E}_i = \{X\mid X\in\Pi_i\mbox{ and } X\cap B_i^p(E)\neq\emptyset\}$.   Then, for each $Y\in\mathcal{E}_i$, we have $P_i(E\mid Y)\geq p$.   Hence,  
    
    $$P_i(E\mid B_i^p(E))=\sum_{Y\in \mathcal{E}_i} P(E\mid Y)\cdot P(Y\mid B_i^p(E))\geq \sum_{Y\in \mathcal{E}_i}  p\cdot P(Y\mid B_i^p(E)) = p\sum_{Y\in \mathcal{E}_i}  P(Y\mid B_i^p(E))=p.$$
        \item Let $w\in CB^p(X)$. 
    Since $p>0$, we have $\Pi_i(w)\cap X\neq \emptyset$ for all $i\in\mathcal{A}$. 
    So $P_{i}(E\mid\Pi_i(w))=r_i$ for all $i\in\mathcal{A}$. 
    Thus, $w\in X$. 
    Therefore,  $CB^p(X)\subseteq X$.     By similar reasoning, we obtain that $P_{i,w}(E)=r_i$ for all $w\in B^p_iCB^p(X)$. 
    Then, we obtain the desired result using reasoning similar to part 3. since $\Pi_i$ partitions $B_i^p CB^p(X)$.
    \end{enumerate}
 \end{proof}

We can now state the Monderer-Samet-Neeman agreement theorem which, for brevity, we call  the {\bf MSN-Theorem}.
\begin{restatable}[MSN-Theorem \cite{MondererSamet1989,Neeman1996}]{theorem}{msntheorem}
\label{MSN-theorem}
    Suppose that $\langle W,(\Pi_i)_{i\in\mathcal{A}},\mathcal{F}, P \rangle$ is an epistemic probability model with a common prior, $p\in (0, 1]$, and for each $i\in \mathcal{A}$,  $r_i\in[0,1]$. For any event $E\subseteq W$, if 
    $$CB^p(\{w \mid P_{i,w}(E)=r_i \text{ for all }i\in\mathcal{A}\})\neq \emptyset$$ 
then $|r_i-r_j| \leq 1-p$ for all $i,j\in\mathcal{A}$.
    
\end{restatable}

The original proof of Monderer, Samet, and Neeman's   uses Lemma \ref{UsefulLemma1} and the following fact about probability measures $P$: For all events $A, B, C$,  if $A\subseteq B\subseteq C$, $P(B) > 0$, and $P(C)> 0$, then $P(A\mid B)\cdot P(B\mid C) = P(A\mid C)$. In the full version of the paper, we give a second proof that avoids explicit mention of multiplication, appealing instead to the following Lemma:   
  
 \begin{restatable}{lemma}{comparisonlemma}
  \label{CP6InStandardProbabilitySetting}
    Suppose that $\langle W,(\Pi_i)_{i\in\mathcal{A}},\mathcal{F}, P \rangle$ is an epistemic probability model with a common prior.   For all events  $E, X\subseteq W$ and $p\in (0, 1]$, if  $CB^p(X)\neq \emptyset$, then  for all $i, j\in \mathcal{A}$, either
 \begin{enumerate}
 \item $P(CB^p(X) - E\mid B^p_i CB^p(X))\leq P(CB^p(X) - E\mid B^p_j CB^p(X))$, and
 \item $P(E\cap CB^p(X)\mid B^p_i CB^p(X)) \leq P(E\cap CB^p(X)\mid B^p_j CB^p(X))$.
 \end{enumerate}
or
 \begin{enumerate}
\item[3.] $P(CB^p(X) - E\mid B^p_i CB^p(X))\geq P(CB^p(X) - E\mid B^p_j CB^p(X))$, and
\item[4.] $P(E\cap CB^p(X)\mid B^p_i CB^p(X)) \geq P(E\cap CB^p(X)\mid B^p_j CB^p(X))$.
\end{enumerate}
\end{restatable}
\begin{proof} 
Suppose that $\langle W,(\Pi_i)_{i\in\mathcal{A}}, P \rangle$ is an epistemic probability model with a common prior.   Let $E, X\subseteq W$ and $p\in (0, 1]$.  We first note that for all $i, j\in\mathcal{A}$, either 
\begin{center}$P(CB^p(X)\mid B^p_iCB^p(X))\geq P(CB^p(X)\mid B^p_jCB^p(X))$ or $P(CB^p(X)\mid B^p_jCB^p(X))\geq P(CB^p(X)\mid B^p_iCB^p(X))$.\end{center}   Without loss of generality, let $0\neq P(CB^p(X)\mid B^p_iCB^p(X))\geq P(CB^p(X)\mid B^p_jCB^p(X))$.   Since $CB^p(X)\subseteq B^pCB^p(X)$, this means that: 
    $$0\neq \frac{P(CB^p(X))}{P(B^p_iCB^p (X))}\geq \frac{P(CB^p(X))}{P(B^p_j CB^p (X))}.$$
This implies that: 
    $$P(B^p_iCB^p(X))\leq P(B^p_jCB^p(X))$$
Thus, for all $F\subseteq B^p_iCB^p(X)\cap B^p_jCB^p(X)$, we have $P(F\mid B^p_iCB^p(X))\geq P(F\mid B^p_jCB^p(X))$.   In particular, this holds for $F=CB^p(X)-E$ and $F=CB^p(X)\cap E$.   This completes the proof.   
\end{proof}

 As we will see in Section \ref{main-theorems}, these different proofs suggest two ways to generalize the  MSN-Theorem. 

\section{Epistemic Plausibility Models}\label{epistemic-plausibility-models}

In this section, we introduce \textbf{epistemic plausibility models}. These models generalize epistemic probability models (Definition \ref{epistemic-probability-model})  by replacing probabilities with the more general notion of a \emph{conditional plausibility measure}.

To define a conditional plausibility measure, we first fix a non-empty set $D$ of \textbf{plausibility values} satisfying the following two conditions. 
First, it is partially ordered by a relation $\leq_D\subseteq D\times D$. 
That is, $\leq_D$ is reflexive (for all $d\in D$, $d\le_D d$), transitive (for all $d,d',d''\in D$, if $d\le_D d'$ and $d'\le_D d''$, then $d\le_D d''$), and antisymmetric (for all $d,d'\in D$, if $d\le_D d'$ and $d'\le_D d$, then $d=d'$). Second, $D$ contains minimal and maximal elements $\bot_D$ and $\top_D$ respectively, such that $\bot_D\le_D d\le_D \top_D$ for all $d\in D$. As usual, we write $d<_D d'$ when $d\le_D d'$ and $d\ne d'$. When $D$ is clear from context, we omit subscripts and write  $\le$, $<$, $\top$, and $\bot$. 

A conditional plausibility measure for an agent $i$ assigns an element from $d\in D$ to a pair of events $X$ and $Y$, denoted $Pl_i(X\mid Y)=d$  meaning that according to $i$, the plausibility of $X$ conditional on $Y$ is $d$.   The domain of a conditional plausibility measure is a {\bf Popper algebra}.   Suppose that $W$ is a nonempty set.   A collection of subsets $\mathcal{F}\subseteq \wp(W)$ (where $\wp(W)$ is the powerset of $W$) is called an {\bf algebra} when $W\in\mathcal{F}$, $\mathcal{F}$ is closed under union (if  $X, Y\in \mathcal{F}$, then $X\cup Y\in\mathcal{F}$), and closed under complement (if $X\in\mathcal{F}$ then $\overline{X}\in\mathcal{F}$, where $\overline{X}$ denotes the complement of $X$ in $W$).  An algebra $\mathcal{F}$ is called a {\bf $\sigma$-algebra} when it is closed under countable unions. 
Then, a set $\mathcal{F}\times\mathcal{F}^\prime$ is a {\bf Popper algebra} if: 
\begin{enumerate}
\item $\mathcal{F}$ is an $\sigma$-algebra; 
\item $\mathcal{F}'\subseteq \mathcal{F}$; and
\item for $X\in\mathcal{F}'$, if $X\subseteq Y$ and $Y\in\mathcal{F}$, then $Y\in\mathcal{F}'$.
\end{enumerate}
\noindent The intended interpretation is that  $\mathcal{F}$ is the set of events that are assigned plausibility values, and $\mathcal{F}^\prime$ is the set of events that can be learned.   A conditional plausibility measure satisfies  very minimal properties: 

\begin{definition}[Conditional Plausibility Structure] Suppose that $D$ is a partially  ordered set  of plausibility values (with ordering $\le_D$).   A {\bf conditional plausibility structure} on $D$ is a tuple $\langle W, \mathcal{F}, \mathcal{F}', (Pl_i)_{i\in\mathcal{A}}\rangle$, where $W\ne\emptyset$, $\mathcal{F}\times \mathcal{F}'$ is a Popper algebra on $W$, and for each $i\in\mathcal{A}$, $Pl_i:\mathcal{F}\times \mathcal{F}'\rightarrow D$ satisfies: For all $X, Y\in \mathcal{F}$ and $E\in\mathcal{F}'$: 
\begin{description}
\item[(CP1)] $Pl_i(E\ | E)=\top$,
\item[(CP2)] $Pl_i(\emptyset\ |\ E)=\bot$,
\item[(CP3)] If $X\subseteq Y$, then $Pl_i(X\ |\ E)\le Pl_i(Y\ |\ E)$, and
\item[(CP4)]  $Pl_i(X\ |\ E) = Pl_i(X\cap E\ |\ E)$.
\end{description} 
\noindent A conditional plausibility structure is {\bf acceptable} if for any $E\in \mathcal{F}$ and $F\in\mathcal{F}^\prime$, we have $E\cap F\in\mathcal{F}^\prime$ whenever $Pl_i(E\mid F)\neq \bot$ for some $i\in\mathcal{A}$. 
A conditional plausibility structure is said to have a {\bf common prior} if $Pl_i=Pl_j$ for all $i,j\in\mathcal{A}$. 
In this case, we  write it as $\langle W, \mathcal{F}, \mathcal{F}^\prime, Pl\rangle$. 
\end{definition}

The properties {\bf CP1}–{\bf CP4} are minimal conditions satisfied by most formal models of belief \cite{sep-formal-belief,Halpern2017}.  Property {\bf CP1} states that any event conditioned on itself receives the maximal plausibility value, while {\bf CP2} requires assigning the minimal plausibility value to the empty set $\emptyset$. Property {\bf CP3} ensures that the plausibility measures are monotonic, and {\bf CP4} captures a minimal condition that any sensible notion of conditional belief should satisfy. {\bf Acceptability} generalizes the following property of standard probability measures $P$: if $P(E\mid F)\neq 0$, then $P(E\cap F)\neq 0$;  and so,  conditioning on $E\cap F$ is well-defined.

We are now ready to define analogues of epistemic probability models (Definition \ref{epistemic-probability-model}), posterior belief (Definition \ref{posterior-belief}), and $p$-belief operators (Definition \ref{p-belief}), replacing probability measures with conditional plausibility measures.   So, for instance, we write $B_i^d(E)$ for the event that $i$ assigns the plausibility value at least $d$ (according to the ordering $\geq_D$) to the event $E$. 

\begin{definition}[Epistemic Plausibility Model]\footnote{Cf. the epistemic probability models defined in \cite{BaltagSmets2008}.}\label{ep-plaus-model} An {\bf epistemic plausibility model} is a tuple $$\langle W, (\Pi_i)_{i\in\mathcal{A}}, \mathcal{F}, \mathcal{F}', (Pl_i)_{i\in\mathcal{A}}\rangle,$$ where  $\langle W,  \mathcal{F}, \mathcal{F}', (Pl_i)_{i\in\mathcal{A}}\rangle$ is a conditional plausibility model such that for all $w\in W$ and $i\in\mathcal{A}$, $\mbox{$\Pi_i(w)\in\mathcal{F}'$}$, and for each $i\in\mathcal{A}$, $\Pi_i$ is a partition of $W$.    We say that an  epistemic plausibility model has a common prior if its conditional plausibility structure has a common prior.

 \end{definition}

\begin{definition}[Posterior Belief] Suppose that $\langle W, (\Pi_i)_{i\in\mathcal{A}}, \mathcal{F}, \mathcal{F}', (Pl_i)_{i\in\mathcal{A}}\rangle$ is an epistemic plausibility model. The {\bf posterior belief  for $i$ at state $w$} is   the function $Pl_{i, w}:\mathcal{F}\rightarrow D$, where,   for all  $X\in\mathcal{F}$,  $$Pl_{i,w}(X)=Pl(X\ |\ \Pi_i(w)).$$
\end{definition}
\begin{remark} For each agent $i\in\mathcal{A}$ and state $w\in W$, the posterior belief function $Pl_{i,w}$ is a \textbf{plausibility measure}: a function $\mu$ from a $\sigma$-algebra to a partially ordered set $D$ satisfying $\mu(W)=\top$ (cf. \textbf{CP1}), $\mu(\emptyset)=\bot$ (cf. \textbf{CP2}), and monotonicity (cf. \textbf{CP3}): if $X\subseteq Y$, then $\mu(X)\leq \mu(Y)$. \end{remark}

 \begin{definition}[$d$-Belief Operator] Suppose that $\langle W, (\Pi_i)_{i\in\mathcal{A}}, \mathcal{F}, \mathcal{F}', (Pl)_{i\in\mathcal{A}}\rangle$ is an epistemic plausibility model for a set $D$ of plausibility values.   For each $d\in D$ and $i\in\mathcal{A}$, a {\bf $d$-belief operator} for $i$ is the function  $B_i^d:\mathcal{F}\rightarrow \wp(W)$, where for all $X\subseteq W$, $$B_i^d(X)=\{w\ |\ Pl_i(X\ |\ \Pi_i(w))\ge d\}.$$
\end{definition}

\begin{remark} In the remainder of this paper, we assume the following: In any epistemic plausibility model $\langle W, (\Pi_i)_{i\in\mathcal{A}}, \mathcal{F}, \mathcal{F}', (Pl_i)_{i\in\mathcal{A}}\rangle$,   for all $E\in\mathcal{F}$, we have that $B_i^d(E)\in\mathcal{F}'$.
\end{remark}
With these definitions in place, the notion of common $p$-belief (Definition \ref{common-p-belief}) extends straightforwardly to epistemic plausibility models; we omit the explicit definition here for brevity.

\section{MSN-Theorems in Epistemic Plausibility Structures}\label{main-theorems}

In this section, we present two distinct generalizations of the MSN-Theorem (Theorem \ref{MSN-theorem}) using epistemic plausibility models (Definition \ref{epistemic-plausibility-models}). Our goal is to identify the minimal properties necessary to prove these generalized versions of the MSN-Theorem.

The first step is to equip plausibility values with a notion of {\em addition}. Indeed, both Aumann's Theorem (Theorem \ref{agreement-thm}) and the MSN-Theorem (Theorem \ref{MSN-theorem}) crucially rely on conditions involving the addition of plausibility values. Formally, this involves specifying a function $\oplus: D\times D\rightarrow D$ that maps each pair $\langle d,f\rangle\in D\times D$ to an element $d\oplus f\in D$, satisfying at least the following minimal conditions:

\begin{definition}[Additive Conditional Plausibility Structure] 
\label{AdditiveConditionalPlausibilityStructure} A conditional plausibility structure\\ $\langle W, \mathcal{F}, \mathcal{F}', (Pl_i)_{i\in\mathcal{A}}\rangle$ is said to be an  {\bf additive conditional plausibility structure} when there is a (partial) function $\oplus:D\times D\to D$ such that for each $i\in\mathcal{A}$,\footnote{If $F$ is a (partial) function, we write $dom(F)$ for the domain of $F$.} 
\begin{description}

\item[(A1)] For all $a,b,c\in D$, if $\langle a, b\rangle\in dom(\oplus)$, then $c\leq b$ iff $\langle a,c\rangle\in dom(\oplus)$ and $a\oplus c\leq a\oplus b$. Similarly, if $\langle b, a\rangle\in dom(\oplus)$, then $c\leq b$ iff  $\langle c,a\rangle\in dom(\oplus)$ and $c\oplus a \leq b\oplus a$. 

\item[(A2)] For all $a\in D$, there exists $b\in D$ such that $a\oplus b=\top$.

\item[(A3)] For all $X,Y\in\mathcal{F}$ and $Z\in \mathcal{F}^\prime$, if $X\cap Y=\emptyset$, then $Pl_i(X\cup Y\mid Z)=Pl_i(X\mid Z)\oplus Pl_i(Y\mid Z)$, 

\item[(A4)] For all $E\in \mathcal{F}$ and sets of disjoint sets $\mathcal{S}\subseteq \mathcal{F}^\prime$, if $a\leq Pl_i(E\mid S) \leq b$ for all $S\in\mathcal{S}$, then $$a\leq Pl_i(E\mid \bigcup \mathcal{S})\leq b.$$

\end{description}

\end{definition}
We say that the epistemic plausibility model $\langle W, (\Pi_i)_{i\in\mathcal{A}}, \mathcal{F}, \mathcal{F}', (Pl_i)_{i\in\mathcal{A}}\rangle$ is {\bf additive} when the  underlying conditional plausibility structure $\langle W, \mathcal{F}, \mathcal{F}', (Pl_i)_{i\in\mathcal{A}}\rangle$ is additive. 

Properties {\bf A1} and {\bf A2} concern the binary operation $\oplus$. Property {\bf A1} states that $\oplus$ is monotonic in both arguments, while {\bf A2} requires each element to have an additive inverse. Properties {\bf A3} and {\bf A4} describe additive properties of the conditional plausibility measures. Specifically, {\bf A3} requires that each conditional plausibility measure $Pl_i$ is finitely additive. Property {\bf A4}  captures a crucial requirement needed to prove the agreement theorems (Theorems \ref{agreement-thm} and \ref{MSN-theorem}).   

\begin{example}  Let $\langle W,\mathcal{F},\mathcal{F}^\prime,Pl\rangle$ be a conditional plausibility structure where: 
\begin{itemize}
    \item $W=\{w_1,w_2,w_3\}$;
    \item $D=\{\frac{n}{6}\mid n\in\{0,1,2,\dots,6\}\}$;
    \item $\oplus=+$;
    \item $\mathcal{F}=\wp(W)$ and $\mathcal{F}^\prime=\wp(W)-\{\emptyset\}$;
    \item For all $E\in \mathcal{F}$ and $F\in\mathcal{F}^\prime$, $Pl(E\mid F)=\frac{\mid E\cap F\mid}{\mid F\mid}$. 
\end{itemize}
\noindent It is not difficult to see that this is an additive conditional plausibility structure by checking that it satisfies {\bf CP1} - {\bf CP4} and {\bf A1} - {\bf A4}. 
\end{example}

There are two important consequences of additive conditional plausibility models. First, we obtain an analogue of Lemma \ref{UsefulLemma1}. The statement  is omitted since it is an easy variation of Lemma \ref{UsefulLemma1}, and its proof is an immediate consequence of \textbf{A4}. The second consequence is that we can introduce a notion of subtraction. This is essential because the MSN-Theorem is stated in terms of differences between plausibility values. The existence of a subtraction is a direct consequence of \textbf{A1}:

\begin{lemma}[Existence of a Subtraction] \label{subtraction}   Suppose that   $\langle W, \mathcal{F}, \mathcal{F}', (Pl_i)_{i\in\mathcal{A}}\rangle$ is an additive conditional plausibility structure on a set $D$ of plausibility values.  For all $a, b \in D$,  if there is a $c\in D$, such that $a\oplus c = b$, then this $c$ is unique.  When such a  $c$ exists, we denoted it by $b\ominus a$.
\end{lemma}

We first note some useful algebraic properties that immediately follow from Lemma \ref{subtraction} and the definition of an additive conditional plausibility structure (Definition \ref{AdditiveConditionalPlausibilityStructure}).

\begin{lemma}
   Suppose that   $\langle W, \mathcal{F}, \mathcal{F}', (Pl_i)_{i\in\mathcal{A}}\rangle$ is an additive conditional plausibility structure on a set $D$ of plausibility values.  Then, 
\begin{enumerate}
    \item For all $a, b\in D$, $\top \ominus a\leq \top\ominus b$ iff $a\geq b$ , 
    \item For all $a, b, c, d\in D$, if $a\leq b\leq c\leq d$ and  both $c \ominus b$ and $d \ominus a$ exist, then $c \ominus b\leq a \ominus d$,
    \item If $\oplus$ is {\bf associative} (for all $a, b, c\in D$, $a\oplus (b\oplus c)=(a\oplus b)  \oplus c$), then for all $a, b\in D$, $$(\top\ominus a) \ominus b=\top \ominus (a\oplus b).$$ \end{enumerate}
\end{lemma}

To motivate our generalizations, we first present a new characterization of the original MSN-Theorem.

\begin{theorem}
\label{CharacterizationOfMSN-Theorem}
    Given an additive conditional plausibility structure with common prior $\mathcal{M}=\langle W,\mathcal{F},\mathcal{F}^\prime,Pl\rangle$, the following are equivalent. 
    \begin{enumerate}
        \item For all additive epistemic plausibility models based on $\mathcal{M}$, if the agents commonly $d$-believe their posteriors of an event $E\in\mathcal{F}$, then their posteriors differ by less than $\top \ominus d$ (if the difference exists). 
        \item For all $A,B\in\mathcal{F^\prime}$, if $Pl(A\mid B)\geq d$ and $Pl(B\mid A)\geq d$, then for any event $E\in\mathcal{F}$, if $\top\ominus d$ exists, then  $$Pl(E\mid A)\ominus Pl(E\mid B)\leq \top \ominus d.$$ 
    \end{enumerate}
\end{theorem}
\noindent Condition 1 is the MSN-theorem. Our strategy to generalize the MSN theorem is to ensure that condition 2 holds for appropriate choices of $A$ and $B$.

Our first generalization of the MSN-Theorem follows the original proof given in \cite{MondererSamet1989} and \cite{Neeman1996}. This proof makes essential use of a property involving the \emph{multiplication} of plausibility values. Therefore, before stating the theorem, we introduce a multiplication operation on the set of plausibility values.

 \begin{definition}[Conditional Plausibility Structure with Addition and Multiplication] 
\label{MutliplicativeConditionalPlausibilityStructure} An additive conditional plausibility structure $\langle W, \mathcal{F}, \mathcal{F}', (Pl_i)_{i\in\mathcal{A}}\rangle$  has a multiplication when there is a (partial) function $\otimes:D\times D\to D$ such that: 
\begin{description}

    \item[({\bf M1})] For all $a,b,c\in D$, if $\langle a, b\rangle\in dom(\otimes)$, $\langle a, c\rangle\in dom(\otimes)$, and  $b\geq c$, then $a\otimes b\geq a\otimes c$ ,
    \item[({\bf M2})] For all $a,b,c\in D$, if $a\oplus b$ is defined, then $(a\oplus b)\otimes c = (a\otimes c)\oplus (b\otimes c)$
    \item[(M3)] For all $i\in\mathcal{A}$, $X\in\mathcal{F} $ and $Y, Z \in \mathcal{F}^\prime$, if $X\subseteq Y\subseteq Z$, then $Pl_i(X\mid Z)=Pl_i(X\mid Y)\otimes Pl_i(Y\mid Z)$ 

\end{description}
 
\end{definition}

Property \textbf{M1} requires that the multiplication operation $\otimes$ is monotonic, while property \textbf{M2} requires that $\otimes$ distributes over the addition operation $\oplus$. The essential property required for the proof is  \textbf{M3}. We can now state our generalized version of the MSN-Theorem.

\begin{restatable}[Generalized MSN-Theorem with Multiplication]{theorem}{GenMSNwithMulti}
    \label{AgreementWithMultiplication} Suppose that  $\langle W, (\Pi_i)_{i\in\mathcal{A}},\mathcal{F},\mathcal{F}^\prime,Pl\rangle$ is an additive epistemic plausibility model with a common prior.  Further, suppose that the model is acceptable, the addition operator $\oplus$ is associative, and there is a  multiplication $\otimes$.   Let $E\in\mathcal{F}$, for each $i\in\mathcal{A}$,  $r_i\in D$,   and  $d\in D$.  If
    $$CB^d(\{w \mid Pl_{i,w}(E)=r_i\text{ for all }i\in\mathcal{A}\})\neq\emptyset$$    
\noindent then for all $i,j\in\mathcal{A}$, if $r_i\ominus r_j$ exists, then $r_i\ominus r_j \leq \top \ominus(\top \otimes d)$.

\end{restatable}
\begin{remark}
    The antecedent of this theorem implies Condition 2 of Theorem \ref{CharacterizationOfMSN-Theorem}.
\end{remark}
 
There is a second generalization of the MSN-Theorem that does not require a multiplication operation. Instead, it relies on two additional properties of an additive conditional plausibility model with a common prior $\langle W, (\Pi_i)_{i\in\mathcal{A}},\mathcal{F},\mathcal{F}^\prime, Pl \rangle$.\footnote{We are numbering them as ({\bf CP6}) and ({\bf CP7}) to avoid confusion with Halpern's ({\bf CP5}) from  \cite{Halpern2001,Halpern2017}.} 
\begin{description}
    \item[({\bf CP6})] For all $E_1,E_2\in\mathcal{F}$ and $A_1,A_2\in \mathcal{F}^\prime$, if $E_1, E_2 \subseteq A_1\cap A_2$ and $\bot\neq Pl(E_1\mid A_1)\geq Pl(E_1\mid A_2)$, then $Pl(E_2\mid A_1)\geq Pl(E_2\mid A_2)$ ,
    \item[({\bf CP7})] For all $E\in\mathcal{F}$ and $i,j\in\mathcal{A}$, either $$Pl(E\mid B_i^p(E))\leq Pl(E\mid B_j^p(E))\mbox{ or }Pl(E\mid B_j^p(E))\leq Pl(E\mid B_i^p(E)).$$
\end{description}
\noindent Property {\bf CP6} captures the following idea: if $i\in\mathcal{A}$, $E_1 \in A_1 \cap A_2$ and $\bot \neq Pl_i(E_1 \mid A_1) \geq Pl_i(E_1 \mid A_2)$, then $A_1$ must be considered ``less plausible'' than $A_2$. Moreover, if $A_1$ is less plausible than $A_2$, then it follows that $Pl(E_2 \mid A_1) \geq Pl(E_2 \mid A_2)$ for all events $E_2 \in A_1 \cap A_2$. Property {\bf CP7} requires that the reliabilities of beliefs held by any two agents can always be compared.   The relationship between {\bf M1}-{\bf M3} and {\bf CP6}-{\bf CP7} is clarified in Appendices \ref{mult-imp-cp6} and \ref{independence-add-mult-axioms}.

We are now ready to state our main theorem: a version of the MSN-Theorem that does not rely on multiplication.

\begin{restatable}[Generalized MSN-Theorem without Multiplication]{theorem}{GenMSNwithoutMulti}
\label{AgreementWithoutMultiplication} Suppose that $\langle W, (\Pi_i)_{i\in\mathcal{A}},\mathcal{F},\mathcal{F}^\prime,Pl\rangle$ is an additive epistemic plausibility model with a common prior that satisfies {\bf CP6} and {\bf CP7}.  Let $E\in\mathcal{F}$, for each $i\in\mathcal{A}$,  $r_i\in D$,   and  $d\in D$.  If
$$CB^d(\{w \mid Pl_{i,w}(E)=r_i \text{ for all }i\in\mathcal{A}\})\neq \emptyset$$ 
then for all $i,j\in\mathcal{A}$, if $ r_i\ominus r_j $ exists, then $(r_i\ominus r_j) \leq (\top \ominus d).$
\end{restatable}
\begin{remark}
    {\bf CP6} - {\bf CP7} guarantee that Condition 2 of Theorem \ref{CharacterizationOfMSN-Theorem} is satisfied for suitable A and B.
\end{remark}
\begin{remark}
\label{Bot+Bot}
The proofs of Theorems \ref{AgreementWithMultiplication} and \ref{AgreementWithoutMultiplication} do not depend on whether condition {\bf A3} holds when $Pl_i(X\mid Z)=Pl_i(Y\mid Z)=\bot$. This observation is useful in some applications.
\end{remark}

Both versions of the MSN-Theorem are important. Some additive epistemic plausibility structures satisfy conditions {\bf CP6} - {\bf CP7} but do not allow for a multiplication operation satisfying {\bf M1} - {\bf M3}. Conversely, other structures have a suitable multiplication operation $\otimes$ satisfying {\bf M1} - {\bf M3} but fail to satisfy conditions {\bf CP6} - {\bf CP7}. Nevertheless, certain natural conditions on multiplication imply property {\bf CP6}. The proofs of both generalizations of the MSN-Theorem are available in the full version of the paper.

\section{Conclusion and Applications}\label{applications}\label{conclusion}

This extended abstract presents two generalizations of the Monderer–Samet–Neeman (MSN) agreement theorem, using additive conditional plausibility measures. These measures provide a unified framework that incorporates many existing formal models of belief. The first generalization closely follows Monderer, Samet, and Neeman's original proof, relying on a multiplication operation, while the second generalization avoids this requirement. Together, these two results clarify the minimal assumptions required to establish MSN-type agreement theorems in various formal belief models.

To illustrate the significance of these assumptions, we briefly examine several existing models of belief (a complete discussion is left to the full version of the paper). Our generalized MSN-theorems apply, in particular, to conditional probability structures \cite{Tsakas2018}, lexicographic probability structures  \cite{BachPerea2013,BachCabessa2023}, and order models of belief \cite{DegremontRoy2012}. When these conditions are not satisfied, the corresponding MSN-theorem typically fails or becomes trivial for the given belief model.

\bibliographystyle{eptcs}
\bibliography{common-p-belief-plausibility-measures}

 \newpage 

\appendix

\section{Multiplication Implies ({\bf CP6})}\label{mult-imp-cp6}

\begin{proposition}
    Given an additive epistemic plausibility model with common prior $\langle W, (\Pi_i)_{i\in\mathcal{A}},\mathcal{F},\mathcal{F}^\prime,Pl\rangle$. If there exists a function $\otimes:D\times D\to D$ that satisfies ({\bf M1}), ({\bf M3}) and 
    \begin{description}
        \item[({\bf M4})] For all $a,b,c\in D$, if $a\neq \bot$ and $a\otimes b\geq a\otimes c$, then 
        $b\geq c$. 
    \end{description}
    Then ({\bf CP6}) is satisfied. 
\end{proposition}
\begin{proof}
    Let $E_1,E_2\in\mathcal{F}$, $A_1,A_2\in\mathcal{F}^\prime$, and $E_1,E_2\subseteq A_1\cap A_2$. 
    Assume that $\bot\neq Pl(E_1\mid A_1)\geq Pl(E_1\mid A_2)$. 
    Given ({\bf M3}), this amounts to
    $$\bot\neq Pl(E_1\mid A_1\cap A_2)\otimes Pl(A_1\cap A_2\mid A_1)\geq Pl(E_1\mid A_1\cap A_2)\otimes Pl(A_1\cap A_2\mid A_2)$$
    Now we have $Pl(E_1\mid A_1\cap A_2)\neq \bot$, since otherwise $Pl(E_1\mid A_1\cap A_2)\otimes Pl(A_1\cap A_2\mid A_1)=Pl(\emptyset\mid A_1\cap A_2)\otimes Pl(A_1\cap A_2\mid A_1)=Pl(\emptyset\mid A_1)=\bot$. 
    Then, by ({\bf M4}), we can obtain
    $$Pl(A_1\cap A_2\mid A_1)\geq Pl(A_1\cap A_2\mid A_2)$$
    which, by ({\bf M1}) implies that 
    $$Pl(E_2\mid A_1)=Pl(E_2\mid A_1\cap A_2)\otimes Pl(A_1\cap A_2\mid A_1)\geq Pl(E_2\mid A_1\cap A_2)\otimes Pl(A_1\cap A_2\mid A_2)=Pl(E_2\mid A_2)$$
\end{proof}

\section{Independence of ({\bf CP6})-({\bf CP7}) and ({\bf M1})-({\bf M3})}\label{independence-add-mult-axioms}

First, we show that even ({\bf M1})-({\bf M4}) do not imply ({\bf CP7}). 

\begin{example}
Given two standard epistemic probability structures with common prior $\langle W,(\Pi_i)_{i\in \{1,2\}},p\rangle$ and $\langle W,(\Pi_i)_{i\in \{1,2\}}, q\rangle$, where
    \begin{itemize}

        \item W=$\{w_1,w_2,w_3\}$
        \item $\Pi_1=\{\{w_1,w_2\},\{w_3\}\}$
        \item $\Pi_2=\{\{w_1\},\{w_2,w_3\}\}$
        \item $p(w_1)=p(w_2)=0.2$, $p(w_3)=0.6$
        \item $q(w_1)=q(w_2)=0.4$, $q(w_3)=0.2$
    \end{itemize}
    Then, we let $D=[0,1]\times[0,1]$, and $\leq_D,\oplus,\otimes,Pl$ be such that: For all $a_1, a_2, b_1, b_2\in [0,1]$ and $A, B\subseteq W$
    $$\langle a_1,b_1\rangle\leq_D\langle a_2,b_2\rangle\iff a_1\leq a_2 \text{ and }b_1\leq b_2$$
    $$\langle a_1, b_1\rangle\oplus\langle a_2,b_2\rangle=\langle a_1+a_2,b_1+b_2\rangle$$
    $$\langle a_1,b_1\rangle\otimes\langle a_2,b_2\rangle=\langle a_1\cdot a_2,b_1\cdot b_2\rangle$$
    $$Pl(A\mid B)=\langle p(A\mid B), q(A\mid B))\rangle$$
\noindent It can be verified that $\langle W, (\Pi_i)_{i\in\mathcal{A}}, Pl\rangle$ is an additive epistemic plausibility structure with a common prior, that $\oplus$ is associative, and that $\otimes$ satisfies ({\bf M1})-({\bf M4}).
Now, suppose that  $d=\langle 0.1,0.1\rangle$ and $E=\{w_2\}$. 
It can be verified that $Pl(E\mid B_1^p(E))=\langle \frac{1}{2},\frac{1}{2}\rangle$ and $Pl(E\mid B_2^d(E))=\langle \frac{1}{4},\frac{2}{3}\rangle$, which violates ({\bf CP7}). 
\end{example}

Now we show that ({\bf CP6})-({\bf CP7}) do not imply that ({\bf M1})-({\bf M3}) can be satisfied. 
\begin{example}
\label{CP67NotImplyMultiplication}
Given a single-agent standard epistemic probability structures with common prior $\langle W,\Pi,p\rangle$ where
    \begin{itemize}
        \item W=$\{w_1,w_2,w_3\}$
        \item $\Pi=\{W\}$
        \item $p(w_1)=p(w_2)=0.25$ and $p(w_3)=0.5$
    \end{itemize}
    Then we define $$p^\prime (A\mid B)=\left\{
    \begin{array}{l l}
     0.26    &\hspace{3mm} \text{if }A=\{w_1\}\text{ and } B=W \\
     0.24    &\hspace{3mm}  \text{if }A=\{w_2\}\text{ and } B=W \\
     0.76& \hspace{3mm} \text{if } A=\{w_1,w_3\}\text{ and }B=W \\
     0.74& \hspace{3mm} \text{if } A=\{w_2,w_3\}\text{ and }B=W \\
      p(A\mid B)   & \hspace{3mm}  \text{otherwise}
    \end{array} \right.$$
    It can be verified that $\langle W,\Pi,p^\prime\rangle$ is an additive epistemic plausibility structure with a common prior that satisfies ({\bf CP6})-({\bf CP7}). 
    However, $p^\prime (\{w_1\}\mid\{w_1,w_2\})=p^\prime (\{w_2\}\mid\{w_1,w_2\})$ but $p^\prime (\{w_1\}\mid\{w_1,w_2,w_3\})\neq p^\prime (\{w_2\}\mid\{w_1,w_2,w_3\})$, which implies that no function $\otimes$ can satisfy ({\bf M3}). 
\end{example}

\end{document}